\documentclass{llncs}
\usepackage{llncsdoc}
\usepackage{graphics}

\spnewtheorem{conjecturee}[conjecture]{Conjecture}{\bfseries}{\itshape}

\begin{document}
\markboth{\LaTeXe{} Class for Lecture Notes in Computer
Science}{\LaTeXe{} Class for Lecture Notes in Computer Science}
\thispagestyle{empty}
\pagestyle{plain}

\title{
The Complexity of Determining  Existence a Hamiltonian Cycle is
 $O(n^3)$ }
\author{Guohun Zhu}
\institute{  Guilin University of Electronic Technology,\newline
            No.1 Jinji Road,Guilin, Guangxi, 541004,P.R.China \newline
            \email{ccghzhu@guet.edu.cn}}

\maketitle

\begin{abstract}
The Hamiltonian cycle problem  in digraph  is mapped into a matching
cover bipartite graph. Based on this mapping, it is proved that
determining existence a Hamiltonian cycle in graph is $O(n^3)$.
\\
\end{abstract}
\begin{abstract}
Hamiltonian Cycle, Z-mapping graph, complexity, decision, matching
covered, optimization
\end{abstract}

\section{Introduction}
It is well known that Hamiltonian cycle problem include two
question.
\begin{problem}
\label{problem1}
 Determining whether a graph has a Hamiltonian
cycle?
\end{problem}

\begin{problem}
\label{problem2}
 Obtaining  a Hamiltonian cycle from a graph if it
exists?
\end{problem}

Both two problem are standard NP-complete problem in general graph
\cite{Johnson1985}. But many research try to solve its in polynomial
constraints the  graph in special cases. Such as semicomplete
multipartite digraphs \cite{Bang-Jensen1998}, solid grid graphs
\cite{William1997}, etc..

 Recently Zhu \cite{zhu2007} had
proved that the complexity of finding Hamiltonian cycle in digraph
with degree bound two existence or not is $O(n^4)$, this complexity
of digraph had been proved $NP-complete$ by y J.Plesn{\'\i}k
\cite{PLESNIK1978}, thus HCP for general digraph is also polynomial.
But that reduce a new problem:
\begin{problem}
Could determining a Hamiltonian graph in polynomial without
obtaining a solution?
\end{problem}


 This paper answer this question by project the approach from \cite{zhu2007} and derived
from matching cover graph $C(G)$ from $D$, it proved that
determining a Hamiltonian cycle existence or not in a digraph $D$ is
equal to finding a edges set $E \subset D$ from  $C(G)$ and $E$ is
strong connected. Since Thus the complexity is deduce only need
$O(n^3)$.

\section{Definition and properties}

 Throughout this paper we consider the finite simple (un)directed
graph $D=(V,A)$ ($G(V,E)$, respectively), i.e. the graph has no
multi-arcs and no self loops. Let $n$ and $m$ denote the number of
vertices $V$ and arcs $A$ (edges $E$, respectively), respectively.

 As conventional, let $|S|$ denote the number of a set $S$.
 The set of vertices $V$ and set of arcs of $A$ of a digraph $D(V,A)$ are denoted by
 $V=\{v_i | 1 \leq i \leq n\}$ and $A=\{a_j | (1 \leq j \leq m) \wedge a_j=<v_i,v_k>, (v_i \neq v_k \in V) \}$
 respectively,
  where $ <v_i,v_k>$ is a arc from $v_i$ to $v_k$ and
a reverse arc is denoted by $\overleftarrow{a_k}=<v_k,v_i> \mbox{ if
 it exists } $.
 Let the out degree of vertex $v_i$ denoted by $d^{+}(v_i)$,
which has the in degree by denoted as $d^{-}(v_i)$ and has the
degree $d(v_i)$ which equals $d^{+}(v_i)+d^{-}(v_i)$. Let the
$N^+(v_i)=\{v_j| <v_i,v_j> \in A \}$, and $N^-(v_i)=\{v_j |
<v_j,v_i> \in A \}$.

Let us define a forward relation $ \bowtie $ between two arcs as
following, $ a_i \bowtie a_j = v_k \: \mbox{iff}\: a_i=<v_i,v_k>
\wedge a_j=<v_k, v_j> $

It is obvious that $|a_i \bowtie a_i|=0$ . A pair of symmetric arcs
$<a_i,a_j>$ are two arcs of a simple digraph if and only if $|a_i
\bowtie a_j|=1 \wedge |a_j \bowtie a_i|=1$.

A {\it cycle } $L$  is a set of arcs $(a_{1},a_{2},\ldots,a_{q})$ in
a digraph $D$, which obeys two conditions:
\begin{enumerate}
\item[c1.] $ \forall a_i \in L,\exists a_j,a_k \in L \setminus \{a_i\},\;  a_i \bowtie a_j \neq a_j \bowtie a_k \in V $
\item[c2.] $ | \bigcup\limits_{a_i \neq a_j \in L } {a_i \bowtie a_j} |=|L|$
\end{enumerate}
   If a cycle $L$ obeys the following conditions, it is a {\it simple cycle}.
\begin{enumerate}
\item[c3.] $\forall L' \subset L $, $L'$ does not satisfy both conditions $c1$  and  $c2$.
\end{enumerate}

 A {\it Hamiltonian cycle $L$} is also a simple cycle of length $n=|V| \geq
 2$ in digraph. A graph that has at least one Hamiltonian cycle is
called a {\it Hamiltonian graph}.

A graph G=$(V;E)$ is bipartite if the vertex set $V$ can be
partitioned into two sets $X$ and $Y$ (the bipartition) such that
$\exists e_i \in E, x_j \in X, \forall x_k \in X \setminus \{x_j\}$,
$(e_i \bowtie x_j \neq \emptyset \rightarrow e_i \bowtie x_k
=\emptyset)$ ($e_i, Y$, respectively). if $|X|=|Y|$, We call that
$G$ is a balanced bipartite graph. A matching $M \subseteq E$ is a
collection of edges such that every vertex of $V$ is incident to at
most one edge of $M$, a matching of balanced bipartite graph is
perfect if $|M| = |X|$. Hopcroft and Karp shows that constructs a
perfect matching of bipartite in $O((m + n)\sqrt(n))$
\cite{Hopcroft1973}. The matching of bipartite has a relation with
neighborhood of $X$.
\begin{theorem}
\cite{Hall1935} \label{HallTheorem} A bipartite graph $G=(X,Y;E)$
has a matching from $X$ into $Y$ if and only if $|N(S)| \geq S$, for
any $S \subseteq X$.
\end{theorem}

\begin{lemma}
\label{simplecycle} A even length of simple cycle consist of two
disjoin perfect matching.
\end{lemma}

Two matrices representation related graphs are defined as follows.

\begin{definition}
\cite{Pearl1973} The incidence matrix $C_{nm}$  of a undirected
graph $G(V;E)$  is a $(0,1)$-matrix with element
\begin{equation}
c_{ij} = \left\{\begin{array}{ll}
                 1, & \mbox{if $v_i \in e_j$;} \\
                 0,      & \mbox{otherwise.}
                \end{array} \right.
\end{equation}
\end{definition}

It is obvious that every column of an incidence matrix has exactly
two $1$ entries.

\begin{definition}
\label{incidencematrixdef} \cite{Pearl1973}
 The incidence
matrix $C_{nm}$  of directed graph $D$ is $(-1,0,1)$-matrix with
element
\begin{equation}
\label{incidencedef} c_{ij} = \left\{\begin{array}{ll}
                 1, & \mbox{if $<v_i,v_i> \bowtie a_j =v_i $;} \\
                 -1, & \mbox{if $a_j \bowtie <v_i,v_i>=v_i $;} \\
                 0, & \mbox{$otherwise $}.
                \end{array} \right.
\end{equation}
\end{definition}

It is obvious that each column of an incidence matrix of digraph has
exactly one $1$ and one $-1$ entries.

\begin{theorem}
\label{rank_theorem} \cite{Pearl1973} The $C$ is the incidence
matrix of a directed graph with $k$ components the rank of $C$ is
given by
\begin{equation}
\label{rankcomponent}
    r(C)=n-k
\end{equation}
\end{theorem}

In order to convince to describe the graph $D$ properties, in this
paper, we denotes the $r(D)=r(C)$.

\section{Z-mapping graphs and Hamiltonian digraph}
Firstly, let us divided the matrix of $C$ into two groups.
\begin{equation}
\label{C_Plusedef}
      C^+=\left\{c_{ij} | c_{ij} \geq 0 \mbox{ otherwise  $0$ }\right \}
\end{equation}
\begin{equation}
\label{C_Minusdef}
      C^-=\left\{c_{ij} | c_{ij} \leq 0  \mbox{ otherwise  $0$ } \right \}
\end{equation}

\subsection{Concepts of Z-mapping graphs}
\begin{definition}
\label{incidencematrixdef} Let $C$ be a incidence matrix of digraph
$D$, the Z-mapping graph of $D$, denoted as $Z(D)$, is defined as a
balanced bipartite graph $G(X,Y:E)$ with a incidence matrix
$F=\left (
    {\begin{array}{c c}
    C^+  \\
    -C^-
    \end{array}}\right)$,
\end{definition}

Since $F$ is bijection, let definition the $F^-1$ as the reverse
mapping from $G$ to $D$.

 The concept of z-mapping graph of a Digraph with degree bound
two was introduced in \cite{zhu2007} which is named {\it Projector
graph}.

It is easy to deduce that the z-mapping graph of $D$ has following
properties.

\begin{lemma}
\label{bipariteofgamma} A Z-mapping graph $G$ of digraph $D$ with
$n$ vertices and $m$ arcs is an balanced bipartite graph $G(X,Y;E)$
with $|X|=n$,$|Y|=n$ and $|E|=m$
\end{lemma}

Since a simple cycle of $D$ is divided into  disjoint edges of $G$,
thus a lemma follows the lemma~\ref{simplecycle}.
\begin{lemma}
\label{Z-mapping1} A Z-mapping graph $G$ of a simple directed cycle
$L$ with $n$ arcs is a perfect matching with disjoint $m$ edges.
\end{lemma}

\subsection{Hamiltonian digraph Vs  Z-mapping graph}
In \cite{zhu2007}, it is presents a  bijection from a digraph $D$
with degree bound two to a balanced bipartite graph $G$, and a
theorem is follows.
\begin{theorem}
\label{zhu_perfectofgamma} \cite{zhu2007}
 Let $G$ be the Z-mapping
graph of a digraph $D$ with degree bound two, a Hamiltonian cycle of
$D$ is equivalent to a perfect match $M$ in $G$ and
$r(f(M)^\prime)=n-1$.
\end{theorem}

Let extends this bijection to general $digraph$.
\begin{lemma}
\label{perfectofL} Let $G$ be the Z-mapping graph of a digraph $D$,
a Hamiltonian cycle of $D$ is equivalent to a perfect match $M$ in
$G$ and $r(F^{-1}(M))=n-1$.
\end{lemma}

\begin{proof}
$\Rightarrow$ Let the digraph $D(V;A)$ with a Hamiltonian cycle $L$,
the incidence matrix of $(V;L)$ is represented by matrix $C^\prime$,
\begin{equation}
\label{perfectofL1}
 r(C^\prime)=r(L)=n-1
\end{equation}

According to lemma~\ref{Z-mapping1}, the Z-mapping graph $G$ of $D$
has a perfect matching $M=F(L)$, thus
\begin{equation}
\label{perfectofL2}
 L=F^{-1}(M)
\end{equation}
Let equation~\ref{perfectofL2} substitute to
equation~\ref{perfectofL1}. Then  $r(F^{-1}(M))=n-1$.

$\Leftarrow$ Let $G(X,Y;E)$ be the Z-mapping graph of the digraph
$D(V,A)$, and $M$ be a perfect matching in $G$. Let $D^\prime(V,L)$
be a sub graph of $D(V,A)$ satisfies that  $L=F^{-1}(M)$ and
$r(L)=n-1$.

Let $D^\prime(V,L)$ be a strong connected digraph. It deduces that
$\forall v_i \in D^\prime $,$d^+(v_i) \geq 1 \wedge d^-(v_i) \geq
1$. Suppose $\exists v_i \in D^\prime$, $d^+(v_i) > 1$ ($d^-(v_i) >
1$ respectively), Since $|M|=n$, it deduces that $|L|=n$.  So
$\forall v_i \in D^\prime$, $d^+(v_i)=d^-(v_i)=1$,  $D^\prime$ is a
Hamiltonian cycle.
\end{proof}

Unfortunately, according to \cite{zhu2007}, even limited the digraph
with degree bound two, finding Hamiltonian cycle in $D$  need visit
all isomorphism perfect matching in $G$ and obtaining a solution. In
another words, we hope that determining the Hamiltonian graph
without to obtaining a solution.

\subsection{Complexity of determining Hamiltonian cycle}

Let $M$ denotes  set of  perfect matching in $G$. Let $C(G)$ define
as follows.
$$C(G) = G \setminus \{e \in E \wedge e \not \in M \}$$.
\begin{remark}
 The connected  $C(G)$ is named matching cover graph in
 \cite{Mkrtchyan2006}, but this paper extends the concept on disconnected graph.
\end{remark}

Obtaining a $C(G)$ from a bipartite graph $G$ is not difficult, the
complexity is follows.
\begin{lemma}
\label{CGcomplexity} The complexity of obtaining $C(G)$ from a
bipartite graph $G$ is $O(n^3)$.
\end{lemma}
\begin{proof}
Since obtaining a maximal matching from bipartite $G$ is $O(n^2)$,
every balance bipartite graph with $2n$ vertices has maximal degree
$n$, repeat $n$ times, can finding all of matching edges in $G$.
Thus the complexity is $O(n^3)$.
\end{proof}

\begin{theorem}
\label{zhu_perfectofgamma} Let $G$ be the Z-mapping graph of a
digraph $D$, if $r(F^{-1}(C(G))=n-1$ then $D$ is Hamiltonian.
\end{theorem}

Then the bijection between digraph $D$ and a Z-mapping graph $G$ can
be extend between the $C(G)$ and $D$.

\begin{theorem}
\label{zhu_perfectofgamma} Let $G$ be the Z-mapping graph of a
digraph $D$, if $r(F^{-1}(C(G))=n-1$ then $D$ is Hamiltonian.
\end{theorem}

\begin{proof}
Since $C(G)=M_1 \cup \ldots \cup M_i \cup \ldots \cup M_q$($q \geq
1$). There are two cases of $q$.

\begin{enumerate}
\item[$q=1$.] Since $G$ has only one perfect matching $M$ and
$M=C(G)$, then $r(F^{-1}(C(G))=r(F^{-1}(M)=n-1$. According to
theorem~\ref{zhu_perfectofgamma}, $D$ is Hamiltonian
\item[$q \geq 2$.]
Suppose $\exists M_i$, $C_i=F^{-1}M_i$ and  $r(C_i)<n-1$, where
$C_i$ is a incidence matrix $n \times n$.

Since $r(M_1 \cup \ldots \cup M_i \cup \ldots \cup M_q)=n-1$, there
exists a set of edges $E=C(G) \cap M_i$, $M_i$ and $E$ are linear
independence. Let $C^\prime=F^{-1}(M_i \cup E)$,  thus $r(C^\prime)
= n-1$, since $C^\prime$ is $n \times m$ matrix ($n \leq m$), and
there are $n \times n$ sub matrix $C^{\prime\prime}$ of  $C^\prime$
satisfies that $r(C^{\prime\prime})=r(C^\prime)=n-1$. Since
$r(F(C^{\prime\prime})) \geq r(F(C^{\prime\prime}))+1=n $, thus the
$n$ edges in $F(C^{\prime\prime})$ is linear independence. So
$F(C^{\prime\prime})$ is a perfect matching in $G$. Thus $D$ is
Hamiltonian.
\end{enumerate}

\end{proof}

\begin{theorem}
\label{complexityHCP}. Determining the HCP problem in graph $D$ is
$O(n^3)$.
\end{theorem}

\begin{proof}
Since every edge in $G$ can be substitute by two symmetric arcs in
$D$, then the digraph is only need considering.
Since obtaining a
perfect matching in bipartite $G$ is $O(n^3)$, obtaining a $C(G)$ is
$O(n^2)$, and the complexity of rank on matrix is $O(n^3)$. Then
Determining the HCP problme is only $O(n^3)$.
\end{proof}

\section{ Conclude remake}

According to the theorem~\ref{complexityHCP}, the complexity of
determining a Hamiltonian cycle existence or not is only $O(n^3)$.
Thus it proved that $P=NP$ again.

Compare with the results in \cite{zhu2007}, the complexity of
obtaining a Hamiltonian cycle is $O(n^4)$. Since the determining
Hamiltonian cycle belongs to decision problem
(problme~\ref{problem1}), which means it only need answer "yes" or
"no" for a given problem, but obtaining a Hamiltonian from graph
belongs to a optimization problem (problem~\ref{problem2}). Since
traveling salesperson problem (TSP) have both decision and
optimization problem, a conjecture follows.

\begin{conjecturee}
\label{Zhu_conjecture} The complexity of optimization TSP is
$O(n^4)$, but the decision TSP is $O(n^3)$.
\end{conjecturee}



\thebibliography{6}
\itemsep=0pt
\bibitem{Johnson1985}
Papadimitriou, C. H. {\it Computational complexity }, in Lawler, E.
L., J. K. Lenstra, A. H. G. Rinnooy Kan, and D. B. Shmoys, eds., {
\it The Traveling Salesman Problem: A Guided Tour of Combinatorial
Optimization }. Wiley, Chichester, UK. (1985), 37--85


\bibitem{William1997}
William Lenhart, Christopher Umans, "Hamiltonian Cycles in Solid
Grid Graphs," focs, 38th Annual Symposium on Foundations of Computer
Science (FOCS '97),  1997, 496

\bibitem{Bang-Jensen1998}
J{\o}rgen Bang-Jensen, Gregory Gutin, Anders Yeo. {\it A polynomial
algorithm for the Hamiltonian cycle problem in semicomplete
multipartite digraphs }. { Journal of Graph Theory} Vol.29, (1998),
111-132.

\bibitem{zhu2007}
Guohun Zhu,{ \it The Complexity of Hamiltonian Cycle Problem in
Digraps with Degree Bound Two is Polynomial Time },
{arXiv:0704.0309v2}, 2007.

\bibitem{PLESNIK1978}
J.Plesn{\'\i}k,{\it The NP-Completeness of the Hamiltonian Cycle
Problem in Planar digraphs with degree bound two}, {Journal
Information Processing Letters}, Vol.8(1978), 199--201

\bibitem{Hopcroft1973}
J.E. Hopcroft and R.M. Karp , An $n^{5/2} $ Algorithm for Maximum
Matchings in Bipartite Graphs. SIAM J. Comput. Vol.2, (1973),
225--231

\bibitem{Hall1935}
P. Hall, On representative of subsets, J. London Math. Soc. 10,
(1935), 26--30

\bibitem{Pearl1973}
Pearl, M, { \it Matrix Theory and Finite Mathematics},{McGraw-Hill},
{New York},(1973), 332--404.

\bibitem{Mkrtchyan2006}
Mkrtchyan, V.V. {\it A note on minimal matching covered graphs }, {
Discrete Mathematics } Vol.36, (2006),452--455
\end{document}